\documentclass[aps,prl,twocolumn]{revtex4-2}

\usepackage[T1]{fontenc}
\usepackage{amsmath,amssymb,mathtools}
\usepackage{amsthm}
\usepackage[colorlinks=true,citecolor=blue,linkcolor=blue,urlcolor=blue]{hyperref}
\usepackage{comment}

\newtheorem{theorem}{Theorem}

\newtheorem{lemma}{Lemma}

\newcommand{\Tr}{\operatorname{Tr}}
\newcommand{\supp}{\operatorname{supp}}
\newcommand{\PPT}{\operatorname{PPT}}

\newcommand{\cF}{\mathcal F}

\newcommand{\EcPPT}{E_{c,\PPT}^{\mathrm{exact}}}

\begin{document}
\title{PPT Entanglement with Correlated Catalysis: Monotones and Irreversibility}
\begin{abstract}
    Quantum catalysts can overcome otherwise impossible quantum state transformations without being consumed, and allowing them to become correlated with the output makes this assistance substantially more powerful. This raises a fundamental question for entanglement theory: which limitations on state manipulation remain when such correlated catalysts are freely available? We answer this question in the positive-partial-transpose (PPT) resource theory, which allows a substantially broader class of operations than local operations and classical communication (LOCC). We identify general conditions under which regularized relative-entropy measures become strongly superadditive, and use them to construct monotones that constrain correlated catalytic PPT transformations without any knowledge of the catalyst. In particular, we prove that the regularized PPT relative entropy is fully additive and strongly superadditive, resolving an open problem in entanglement theory. Most importantly, these constraints show that even arbitrary correlated catalysts cannot restore asymptotic reversibility: for an explicit state, the optimal entanglement distillation rate remains strictly smaller than the entanglement cost. Thus, substantial catalytic assistance does not remove some of the fundamental limitations of mixed-state entanglement manipulation.
\end{abstract}
\author{Jingsong Ao}
\author{Aby Philip}
\author{Alexander Streltsov}
\affiliation{Institute of Fundamental Technological Research, Polish Academy of Sciences, Pawi\'{n}skiego 5B, 02-106 Warsaw, Poland}
\maketitle

\paragraph{Introduction.}
Quantum entanglement is a central resource in quantum information processing. A basic goal of entanglement theory is not only to determine whether a state is entangled, but also to quantify the amount of entanglement it contains and to characterize which state transformations are possible under a prescribed class of operations \cite{VedralEtAl1997,VedralPlenio1998,Vidal2000,qe,qr}. Resource monotones provide a natural tool for both purposes: because they cannot increase under the allowed operations, they impose necessary constraints on state conversion and quantify limitations on resource manipulation \cite{Vidal2000,HorodeckiLimits2000,BrandaoGour2015}. Two particularly useful properties concern how a resource monotone behaves when systems are combined. For independent systems, one may require that the resources of the subsystems to add exactly to the totoal resource; when correlations are present, one may require the total resource to remain at least as large as the sum of the resources of the individual subsystems. These properties are known as additivity and strong superadditivity, respectively \cite{DonaldHorodeckiRudolph2002}.

Catalysts provide an additional way to assist quantum state transformations. A catalyst is an auxiliary quantum state that can enable an otherwise impossible transformation while being recovered at the end of the process \cite{JonathanPlenio1999,KondraDattaStreltsov2021}. Correlated catalysis allows greater freedom: the catalyst must still be recovered, but it may become correlated with the transformed system \cite{cata,FundamentalLimitsCorrelatedCatalysis}. This enlarges the class of possible catalytic transformations. But, in general, the catalyst enabling a given transformation may not be known in advance. Therefore, it is especially useful to identify resource monotones that depend only on the initial and final states of the system and can rule out transformations without specifying the catalyst.

A natural question is then what properties allow a resource monotone to provide such catalyst-independent constraints. A recent result~\cite{FundamentalLimitsCorrelatedCatalysis} gives a simple answer: under mild continuity assumptions, additivity and strong superadditivity are sufficient for a resource monotone to remain monotone under correlated catalysis. Before the transformation, additivity allows the resource of the independent system and catalyst to be evaluated separately. After the transformation, strong superadditivity provides a lower bound in terms of the system and catalyst even when they have become correlated. Since the catalyst is returned in the same state, its contribution cancels, leaving a constraint that depends only on the initial and final states of the system.

Resource monotones with all of the properties required above are, however, remarkably scarce. Even for entanglement manipulation under local operations and classical communication (LOCC), only a few examples are known,such as the squashed entanglement~\cite{ChristandlWinter} and conditional entanglement of mutual information~\cite{YangHorodeckiWang2008}. This raises a broader question: rather than searching for such monotones one by one, can one identify general conditions that guarantee strong superadditivity? We answer this question for a broad class of regularized relative-entropy measures. We identify general structural conditions under which strong superadditivity follows. Together with known additivity results \cite{BeigiRubboliTomamichel2026,FangFawziFawzi2026}, this provides a systematic route for constructing resource monotones that remain valid under correlated catalysis.

We develop this idea in the resource theory based on positive partial transpose (PPT). A bipartite state is PPT if it remains positive under partial transposition \cite{Peres1996,HorodeckiEtAl1996}, and the corresponding free operations are completely PPT-preserving maps. These operations form a larger class than LOCC, so impossibility results established in this more permissive setting also apply to LOCC transformations. PPT-based quantities have moreover played a central role in bounding entanglement distillation and entanglement cost. Yet, despite this long history, no systematic family of additive and strongly superadditive monotones was known in this setting.

We apply the aforementioned conditions to PPT entanglement using the hierarchy introduced by Wang et al., in \cite{WangJingZhu2025} and further studied in~\cite{FangFawziFawzi2026}. The basic idea is to replace the PPT states by a sequence of larger reference sets that are easier to handle mathematically while still retaining information about PPT entanglement. We denote these sets by $\PPT_k$; they contain all PPT states, and the first level $\PPT_1$ coincides with the Rains set \cite{Rains1999}. At every level $k$, the conditions applies and produces a regularized relative-entropy monotone that is additive and strongly superadditive. Therefore, we obtain a family of monotones that can constrain correlated catalytic transformations.

The hierarchy also allows us to resolve an open question about the regularized PPT relative entropy. This quantity is known to be weakly additive: when independent systems are prepared in the same state, the resource assigned to the joint system is simply the sum of the resources of the individual systems \cite{EisertAudenaertPlenio2003}. What remained open was full additivity, where the independent systems may be in different states, as well as strong superadditivity when correlations are present. We consider the set $\PPT_\infty$ and prove that its regularized relative entropy is exactly equal to the regularized PPT relative entropy. This identity establishes full additivity and strong superadditivity of the regularized PPT relative entropy for arbitrary bipartite states, and hence its monotonicity under correlated catalysis.

After we obtain a family of monotones that can constrain correlated catalytic transformations, we turn our attention to the reversibility. Reversibility is a basic question in asymptotic entanglement manipulation. Starting from many copies of a state, one may first distill maximally entangled states and then use them to prepare the original state again. The process is reversible if the amount that can be distilled matches the amount required for preparation; a strict gap between the distillation rate and the entanglement cost means that some entanglement is inevitably lost. Whether this loss can be removed by enlarging the allowed operations or by supplying additional resources has therefore been studied extensively. Lami and Regula showed that irreversibility persists even under the broad class of non-entangling operations, and later established a strong form of irreversibility under dually non-entangling operations \cite{LamiRegula2023,LamiRegulaDNE2024}. Interestingly, sufficiently general probabilistic transformations can instead restore reversibility \cite{RegulaLami2024}, showing that the answer depends sensitively on what additional freedom is allowed. For catalytic assistance, Lami, Regula, and Streltsov proved that even correlated catalysts cannot restore reversibility under LOCC: PPT bound entangled states remain undistillable while having nonzero entanglement cost \cite{LamiRegulaStreltsov2024}. This argument, however, does not establish irreversibility in the PPT resource theory, because PPT states are free in that setting.

We show that irreversibility nevertheless persists under completely PPT-preserving operations even when arbitrary correlated catalysts are allowed. Using the state introduced by Wang and Duan to demonstrate PPT irreversibility \cite{WangDuan2017}, we show that finite-level monotone from our hierarchy gives an upper bound on the correlated-catalytic distillation rate, while the regularized PPT relative entropy gives a lower bound on the correlated-catalytic entanglement cost. Both bounds coincide with the corresponding rates achievable without a catalyst. The optimal distillation rate therefore remains strictly smaller than the entanglement cost: allowing arbitrary correlated catalysts does not remove the irreversibility of PPT entanglement manipulation.

\paragraph{Definitions.} Let \(A\) and \(B\) denote finite-dimensional quantum systems held by Alice and Bob, respectively, with associated Hilbert spaces \(\mathcal H_A\) and \(\mathcal H_B\). The composite system \(AB\) is associated with the Hilbert space \(\mathcal H_{AB}:=\mathcal H_A\otimes\mathcal H_B\). We denote by \(\mathcal L(AB):=\mathcal L(\mathcal H_{AB})\) the space of linear operators on \(\mathcal H_{AB}\), and write
\begin{equation}
\mathcal D(AB)
:=
\left\{
\rho_{AB}\in\mathcal L(AB):
\rho_{AB}\geq0,\ 
\Tr\rho_{AB}=1
\right\}
\end{equation}
for the set of density operators of the composite system \(AB\). The set of states with positive partial transpose (PPT) is
\begin{equation}
\PPT(A:B)
:=
\left\{
\sigma_{AB}\in\mathcal D(AB):
\sigma_{AB}^{\Gamma_B}\geq0
\right\},
\end{equation}
where \(\Gamma_B\) denotes the partial transpose on subsystem \(B\) with respect to a fixed orthonormal basis of \(\mathcal H_B\)~\cite{Peres1996,HorodeckiEtAl1996}.

A quantum channel \(\Lambda_{AB\to A'B'}\) is completely PPT-preserving if \(\Gamma_{B'}\circ\Lambda\circ\Gamma_B\) is completely positive.

We consider the sequence of comparison sets \(\{\PPT_k(A:B)\}_{k\geq1}\) introduced by Wang, Jing, and Zhu \cite{WangJingZhu2025}. The first member is the Rains set
\cite{Rains99},
\begin{equation}
\PPT_1(A:B)
:=
\left\{
\tau_{AB}\geq0:
\left\|\tau_{AB}^{\Gamma_B}\right\|_1\leq1
\right\}.
\end{equation}
For \(k\geq2\), $\PPT_{k}(A:B)$ are defined recursively by \cite{WangJingZhu2025}
\begin{multline}
\PPT_{k}(A:B)
:=
\bigl\{\tau_{AB}\geq0:\;
\exists\,\eta_{AB}\in\PPT_{k-1}(A:B)\\\text{such that }
-\eta_{AB}\leq\tau_{AB}^{\Gamma_B}\leq\eta_{AB}
\bigr\}.
\end{multline}
These sets satisfy
\cite{WangJingZhu2025}
\begin{equation}
\PPT(A:B)
\subseteq
\PPT_{k+1}(A:B)
\subseteq
\PPT_k(A:B).
\label{eq:pptk-nesting}
\end{equation}

We define
\begin{equation}
\PPT_\infty(A:B)
:=
\bigcap_{k\geq1}\PPT_k(A:B).
\label{eq:ppt-infty-main}
\end{equation}

For \(\rho\in\mathcal D(AB)\) and a positive semidefinite operator \(\tau\), the relative entropy is defined as \cite{MuellerHermesReeb2017}
\begin{equation}
D(\rho\Vert\tau)
:=
\begin{cases}
\Tr[\rho(\log_2\rho-\log_2\tau)],
&\supp\rho\subseteq\supp\tau,\\
+\infty,
&\text{otherwise}.
\end{cases}
\end{equation}
The \(\PPT_k\) relative entropy of entanglement is
\cite{WangJingZhu2025,FangFawziFawzi2026}
\begin{equation}
D_{\PPT_k}(\rho)
:=
\min_{\tau\in\PPT_k(A:B)}
D(\rho\Vert\tau),
\end{equation}
and its regularization is
\cite{FangFawziFawzi2026}
\begin{equation}
D_{\PPT_k}^{\infty}(\rho)
:=
\lim_{n\to\infty}
\frac1n
D_{\PPT_k(A^n:B^n)}
\bigl(\rho^{\otimes n}\bigr).
\end{equation}
We also consider the regularized PPT relative entropy of entanglement
\cite{AudenaertEtAl2002},
\begin{equation}
D_{\PPT}^{\infty}(\rho)
:=
\lim_{n\to\infty}
\frac1n
\min_{\sigma_n\in\PPT(A^n:B^n)}
D\bigl(\rho^{\otimes n}\Vert\sigma_n\bigr).
\end{equation}

We say that a state \(\rho_{AB}\) can be transformed into
\(\sigma_{A'B'}\) with a correlated catalyst if, for every
\(\varepsilon>0\), there exist a bipartite catalyst state
\(\tau_{CD}\), with \(A A' C\) held by Alice and \(BB'D\) by Bob, and a
completely PPT-preserving channel~\cite{cata,LamiRegulaStreltsov2024}
\[
\Lambda:AC:BD\longrightarrow A'C:B'D
\]
such that, for
\[
\omega_{A'B'CD}
=
\Lambda(\rho_{AB}\otimes\tau_{CD}),
\]
one has
\[
\omega_{CD}=\tau_{CD},
\qquad
\frac12\left\|\omega_{A'B'}-\sigma_{A'B'}\right\|_1
\leq \varepsilon .
\].

\paragraph{Main Results.} With the definitions out of way, we can now proceed to the results. 

\begin{theorem}\label{thm:reg-ppt-k-corr-catal}
For every finite \(k\), the regularized \(\PPT_k\)-relative entropy of entanglement \(D_{\PPT_k}^{\infty}\) is monotone under correlated catalytic transformations implemented by completely PPT-preserving maps.
\end{theorem}

Finding resource monotones that remain valid under correlated catalysis is highly nontrivial. Such a resource monotone is required to be additive on tensor-product states, strongly superadditive on arbitrary correlated states, and lower semicontinuous \cite{FundamentalLimitsCorrelatedCatalysis}. Very few resource measures are known to satisfy these three properties simultaneously.

To the best of our knowledge, no such monotone was previously known for the resource theory in which PPT states and completely PPT-preserving maps are free. Even in entanglement theory under LOCC, only two examples are known: the squashed entanglement \cite{ChristandlWinter} and the conditional entanglement of mutual information \cite{YangHorodeckiWang2008}. No separation between these two quantities is known, and whether they coincide on all quantum states remains an open problem.

Theorem~\ref{thm:reg-ppt-k-corr-catal} therefore provides a family of correlated-catalytic monotones for PPT entanglement theory. Since every LOCC map is completely PPT-preserving, all the quantities \(D_{\PPT_k}^{\infty}\) are also monotone under correlated catalytic LOCC transformations.

The different ingredients entering Theorem~\ref{thm:reg-ppt-k-corr-catal} have distinct origins. Once the relevant tensor conditions for \(\PPT_k\) are verified, tensor-product additivity follows directly from Ref.~\cite[Proposition~9]{BeigiRubboliTomamichel2026}. Asymptotic
continuity, which in particular implies lower semicontinuity, follows from the standard continuity argument for relative-entropy distances to convex sets of positive operators~\cite{Winter2016}. The main new ingredient is strong superadditivity on arbitrary correlated states. We establish it through the following general criterion.

\begin{lemma}\label{lem:superadditivity}
Let \(\mathcal F\) assign to every finite-dimensional bipartite
system \(A:B\) a nonempty compact convex set
\(\mathcal F(A:B)\) of subnormalized positive semidefinite operators. Suppose that:

\textnormal{(i)}
\(\mathcal F(A:B)\) contains a full-rank state;

\textnormal{(ii)}
for every \(n\geq1\), the set \(\mathcal F(A^n:B^n)\) is invariant
under permutations of the \(n\) copies of \(AB\);

\textnormal{(iii)}
for arbitrary bipartite systems \(A:B\) and \(A':B'\),
\begin{equation*}
\mathcal F(A:B)\otimes\mathcal F(A':B')
\subseteq
\mathcal F(AA':BB');
\end{equation*}

\textnormal{(iv)}
for arbitrary bipartite systems \(A:B\) and \(A':B'\),
\begin{equation*}
\mathcal F(A:B)^{\circ}\otimes
\mathcal F(A':B')^{\circ}
\subseteq
\mathcal F(AA':BB')^{\circ},
\end{equation*}
where the positive polar of a set \(\mathcal F\) is defined as
\begin{equation*}
\mathcal F^{\circ}
:=
\left\{
X\geq0:
\Tr(X\tau)\leq1
\ \text{for all }\tau\in\mathcal F
\right\}.
\end{equation*}

Then the regularized relative entropy associated with
\(\mathcal F\) is strongly superadditive. Namely, for every
\(\omega_{AA'BB'}\in\mathcal D(AA'BB')\),
\begin{equation*}
D_{\mathcal F}^{\infty}(\omega_{AA'BB'})
\geq
D_{\mathcal F}^{\infty}(\omega_{AB})
+
D_{\mathcal F}^{\infty}(\omega_{A'B'}).
\end{equation*}
\end{lemma}

The proof of Lemma~\ref{lem:superadditivity} is given in the supplemental material. For every finite \(k\), the comparison sets \(\PPT_k\) satisfy conditions~\textnormal{(i)}--\textnormal{(iv)}. Their tensor-product closure and the tensor-product closure of their positive polars were established in Refs.~\cite{FangFawziFawzi2026,BeigiRubboliTomamichel2026}, while the remaining conditions follow directly from their definition. Lemma~\ref{lem:superadditivity} therefore implies that \(D_{\PPT_k}^{\infty}\) is strongly superadditive for every finite \(k\). A further application to the resource theory of magic is given in the supplemental material.

We now turn to proving that \(D_{\PPT}^{\infty}\) is additive under tensor products and strongly superadditive for general states. The key step is to identify it with the regularized relative entropy associated with the comparison set \(\PPT_\infty\). As shown in Eq.~\eqref{eq:pptk-nesting}, the comparison sets become increasingly restrictive as \(k\) increases, while continuing to contain all PPT states. This motivates considering the regularized relative entropy associated with \(\PPT_\infty\).

\begin{theorem}\label{thm:ppt_inf_equal_ppt}
For every finite-dimensional bipartite state
\(\rho_{AB}\in\mathcal D(AB)\),
\begin{equation}
D_{\PPT_\infty}^{\infty}(\rho_{AB})
=
D_{\PPT}^{\infty}(\rho_{AB}).
\label{eq:ppt-infty-main-identity}
\end{equation}
Furthermore, \(D_{\PPT}^{\infty}\) is additive on tensor-product states and strongly superadditive on arbitrary correlated states. Consequently, it is monotone under correlated catalytic transformations implemented by completely PPT-preserving maps.
\end{theorem}

This result is nontrivial because the positive polar of the PPT set is not closed under tensor products, preventing the existing criterion from establishing either additivity on product states or superadditivity on correlated states. The set \(\PPT_\infty\), however, satisfies the required tensor and positive-polar tensor conditions. This allows the general criterion to be applied to \(D_{\PPT_\infty}^{\infty}\), then Eq.~\eqref{eq:ppt-infty-main-identity} transfers the resulting properties to \(D_{\PPT}^{\infty}\).

The monotones established above can be used to address a basic question in asymptotic entanglement manipulation: whether correlated catalysts can restore reversibility. Given many copies of a state, one may distill maximally entangled states from them or, conversely, consume maximally entangled states to prepare the state. Asymptotic manipulation is reversible when the optimal distillation rate equals the entanglement cost, and irreversible when these two quantities differ. Since correlated catalysts enlarge the set of possible transformations while being returned at the end, it is natural to ask whether they can make the two rates coincide.

Let \(E_{d,\PPT}^{\rm cc}\) and \(E_{c,\PPT}^{\rm cc}\) denote, respectively, the optimal asymptotic distillation rate and entanglement cost under completely PPT-preserving protocols assisted by correlated catalysts. Precise definitions are given in the supplemental material.

\begin{theorem}\label{thm:cc-ppt-irreversibility}
Asymptotic entanglement manipulation under completely PPT-preserving operations remains irreversible in the presence of arbitrary correlated catalysts. In particular, let
\begin{equation}
\begin{aligned}
\rho_v&=\frac12 |v_1\rangle\!\langle v_1|
       +\frac12 |v_2\rangle\!\langle v_2|,\\
|v_1\rangle&=\frac{|01\rangle-|10\rangle}{\sqrt2},\\
|v_2\rangle&=\frac{|02\rangle-|20\rangle}{\sqrt2}.
\end{aligned}
\end{equation}
Then
\begin{equation}
E_{d,\PPT}^{\rm cc}(\rho_v)
=
\log_2\!\left(1+\frac1{\sqrt2}\right)
<
1
=
E_{c,\PPT}^{\rm cc}(\rho_v).
\label{eq:cc-irreversibility-main}
\end{equation}
\end{theorem}
Wang and Duan previously established a strict gap between PPT entanglement distillation and entanglement cost without catalytic assistance \cite{WangDuan2017}. More broadly, the fate of reversibility is known to depend sensitively on how far the allowed operations are enlarged. Brand\~ao and Plenio proposed a reversible entanglement theory under asymptotically non-entangling operations, with the regularized relative entropy of entanglement governing asymptotic conversion \cite{BrandaoPlenio2008,BrandaoPlenio2010}. The original argument relied on the generalized quantum Stein lemma \cite{BrandaoPlenioStein2010}, whose proof was later found to contain a gap \cite{BertaEtAl2023}; this lemma and the associated second-law framework have since been re-established \cite{HayashiYamasaki2025}. By contrast, Lami and Regula showed that entanglement remains irreversible under non-entangling operations \cite{LamiRegula2023}. More recently, Lami, Regula, and Streltsov showed that irreversibility under LOCC persists even in the presence of correlated catalysts \cite{LamiRegulaStreltsov2024}. Their argument, however, relies on PPT bound entangled states that remain undistillable while having nonzero entanglement cost under correlated-catalytic LOCC. It therefore does not establish correlated-catalytic irreversibility under completely PPT-preserving operations, for which PPT states are free and have zero entanglement cost. Theorem~\ref{thm:cc-ppt-irreversibility} shows that irreversibility nevertheless persists in this different operational setting even when arbitrary correlated catalysts are allowed.

The proof also illustrates why it is useful to have more than one catalytic monotone. The general correlated-catalytic rate bounds of Ref.~\cite{LamiRegulaStreltsov2024} allow \(D_{\PPT_1}^{\infty}\) to control distillation and \(D_{\PPT}^{\infty}\) to control formation. For the state \(\rho_v\), the ordinary PPT rates and these two relative-entropy quantities are known exactly \cite{WangDuan2017}. Consequently,
\begin{align}
E_{d,\PPT}(\rho_v)
&\leq E_{d,\PPT}^{\rm cc}(\rho_v)
\leq D_{\PPT_1}^{\infty}(\rho_v)
=E_{d,\PPT}(\rho_v),\\
E_{c,\PPT}(\rho_v)
&\geq E_{c,\PPT}^{\rm cc}(\rho_v)
\geq D_{\PPT}^{\infty}(\rho_v)
=E_{c,\PPT}(\rho_v).
\end{align}
The outer inequalities simply use the fact that a trivial catalyst is always allowed, while the inner inequalities are the catalytic monotone bounds~\cite{LamiRegulaStreltsov2024}. Both preceding theorems enter essentially in this argument. Theorem~\ref{thm:reg-ppt-k-corr-catal}, specialized to $k=1$, yields the correlated-catalytic bound \(D_{\PPT_1}^{\infty}\) on distillation, whereas Theorem~\ref{thm:ppt_inf_equal_ppt} yields the corresponding \(D_{\PPT}^{\infty}\) bound on formation. Together, these two bounds establish the irreversibility in Eq.~\eqref{eq:cc-irreversibility-main}. The details are given in the supplemental material.

Moreover, taking nontrivial convex combinations of the squashed entanglement with each of \(D_{\PPT_k}^{\infty}\), \(k\geq1\), and \(D_{\PPT}^{\infty}\) further yields a family of faithful, additive, and strongly superadditive LOCC entanglement monotones.
\paragraph{Discussion.}
We have established a general criterion for strong superadditivity of regularized relative-entropy measures and applied it to PPT entanglement. This yields correlated-catalytic monotones from the \(\PPT_k\) hierarchy, establishes full additivity and strong superadditivity of \(D_{\PPT}^{\infty}\), and shows that PPT entanglement manipulation remains irreversible even with arbitrary correlated catalysts.

An important open question is whether the monotones \(D_{\PPT_k}^{\infty}\), \(k\geq1\), together with \(D_{\PPT}^{\infty}\), form a complete family of monotones for correlated catalytic transformations under completely PPT-preserving maps. More broadly, Lemma~\ref{lem:superadditivity} provides a general mechanism for constructing additive and strongly superadditive regularized relative-entropy measures from comparison sets satisfying suitable conditions. It remains to identify and classify further comparison sets with these properties, which may yield large new families of correlated-catalytic monotones in other quantum resource theories.

Another natural direction is to ask whether probabilistic transformations can restore reversibility in the PPT resource theory. Probabilistic protocols can restore reversibility under asymptotically resource-non-generating operations \cite{RegulaLami2024}. Whether probabilistic transformations can also restore reversibility when the allowed operations are required to be completely PPT-preserving remains open. An especially interesting question is how probabilistic transformations should be combined with catalysis. One natural formulation requires the catalyst to be recovered with certainty in every branch of the protocol, while only the transformation of the system is allowed to fail with some probability. Determining whether PPT entanglement becomes reversible under such probabilistic correlated-catalytic transformations would clarify how much additional freedom is actually required to overcome the irreversibility established here.

\emph{Acknowledgments.} This work was supported by the National Science Centre Poland (Grant No. 2022/46/E/ST2/00115 and 2024/55/B/ST2/01590). \emph{AI~Use~Disclosure.} During the preparation of this manuscript, the authors used OpenAI's ChatGPT 5.6 as an assistive tool for language editing, manuscript organization, bibliographic checking, and the exploration and verification of mathematical arguments. All AI-assisted suggestions were critically assessed by the authors, and the final proofs, references, statements, and conclusions were independently verified by the authors, who take full responsibility for the content of the manuscript.

\bibliography{refs}

@article{FundamentalLimitsCorrelatedCatalysis,
  title = {Fundamental Limits on Correlated Catalytic State Transformations},
  author = {Rubboli, Roberto and Tomamichel, Marco},
  journal = {Physical Review Letters},
  volume = {129},
  issue = {12},
  pages = {120506},
  numpages = {7},
  year = {2022},
  month = {Sep},
  publisher = {American Physical Society},
  doi = {10.1103/PhysRevLett.129.120506},
  url = {https://link.aps.org/doi/10.1103/PhysRevLett.129.120506}
}

@article{qe,
  title = {Quantum entanglement},
  author = {Horodecki, Ryszard and Horodecki, Pawe\l{} and Horodecki, Micha\l{} and Horodecki, Karol},
  journal = {Reviews of Modern Physics},
  volume = {81},
  issue = {2},
  pages = {865--942},
  numpages = {0},
  year = {2009},
  month = {Jun},
  publisher = {American Physical Society},
  doi = {10.1103/RevModPhys.81.865},
  url = {https://link.aps.org/doi/10.1103/RevModPhys.81.865}
}

@article{qr,
  title = {Quantum resource theories},
  author = {Chitambar, Eric and Gour, Gilad},
  journal = {Reviews of Modern Physics},
  volume = {91},
  issue = {2},
  pages = {025001},
  numpages = {48},
  year = {2019},
  month = {Apr},
  publisher = {American Physical Society},
  doi = {10.1103/RevModPhys.91.025001},
  url = {https://link.aps.org/doi/10.1103/RevModPhys.91.025001}
}

@article{cata,
   title={Catalysis of entanglement and other quantum resources},
   volume={86},
   ISSN={1361-6633},
   url={http://dx.doi.org/10.1088/1361-6633/acfbec},
   DOI={10.1088/1361-6633/acfbec},
   number={11},
   journal={Reports on Progress in Physics},
   publisher={IOP Publishing},
   author = {Datta, Chandan and Kondra, Tulja Varun
          and Miller, Marek and Streltsov, Alexander},
   year={2023},
   month=Oct, pages={116002} }

@article{YangHorodeckiWang2008,
  title = {An Additive and Operational Entanglement Measure: Conditional Entanglement of Mutual Information},
  author = {Yang, Dong and Horodecki, Micha\l{} and Wang, Z. D.},
  journal = {Physical Review Letters},
  volume = {101},
  issue = {14},
  pages = {140501},
  numpages = {4},
  year = {2008},
  month = {Sep},
  publisher = {American Physical Society},
  doi = {10.1103/PhysRevLett.101.140501},
  url = {https://link.aps.org/doi/10.1103/PhysRevLett.101.140501}
}

@article{ChristandlWinter,
   title={{“Squashed entanglement”: An additive entanglement measure}},
   volume={45},
   ISSN={1089-7658},
   url={http://dx.doi.org/10.1063/1.1643788},
   DOI={10.1063/1.1643788},
   number={3},
   journal={Journal of Mathematical Physics},
   publisher={AIP Publishing},
   author={Christandl, Matthias and Winter, Andreas},
   year={2004},
   month=Mar, pages={829–840} }

@article{WangJingZhu2025,
  title = {Computable and Faithful Lower Bound on Entanglement Cost},
  author = {Wang, Xin and Jing, Mingrui and Zhu, Chengkai},
  journal = {Physical Review Letters},
  volume = {134},
  issue = {19},
  pages = {190202},
  numpages = {8},
  year = {2025},
  month = {May},
  publisher = {American Physical Society},
  doi = {10.1103/PhysRevLett.134.190202},
  url = {https://link.aps.org/doi/10.1103/PhysRevLett.134.190202}
}

@article{MuellerHermesReeb2017,
   title={Monotonicity of the Quantum Relative Entropy Under Positive Maps},
   volume={18},
   ISSN={1424-0661},
   url={http://dx.doi.org/10.1007/s00023-017-0550-9},
   DOI={10.1007/s00023-017-0550-9},
   number={5},
   journal={Annales Henri Poincaré},
   publisher={Springer Science and Business Media LLC},
   author={Müller-Hermes, Alexander and Reeb, David},
   year={2017},
   month=Jan, pages={1777–1788} }

@article{Peres1996,
   title={Separability Criterion for Density Matrices},
   volume={77},
   ISSN={1079-7114},
   url={http://dx.doi.org/10.1103/PhysRevLett.77.1413},
   DOI={10.1103/physrevlett.77.1413},
   number={8},
   journal={Physical Review Letters},
   publisher={American Physical Society (APS)},
   author={Peres, Asher},
   year={1996},
   month=Aug, pages={1413–1415} }

@article{HorodeckiEtAl1996,
   title={Separability of mixed states: necessary and sufficient conditions},
   volume={223},
   ISSN={0375-9601},
   url={http://dx.doi.org/10.1016/S0375-9601(96)00706-2},
   DOI={10.1016/s0375-9601(96)00706-2},
   number={1-2},
   journal={Physics Letters A},
   publisher={Elsevier BV},
   author={Horodecki, Michał and Horodecki, Paweł and Horodecki, Ryszard},
   year={1996},
   month=Nov, pages={1–8} }

@article{Rains99,
   title={Rigorous treatment of distillable entanglement},
   volume={60},
   ISSN={1094-1622},
   url={http://dx.doi.org/10.1103/PhysRevA.60.173},
   DOI={10.1103/physreva.60.173},
   number={1},
   journal={Physical Review A},
   publisher={American Physical Society (APS)},
   author={Rains, E. M.},
   year={1999},
   month={July}, pages={173–178} }

@article{AudenaertEtAl2002,
  title = {Asymptotic relative entropy of entanglement for orthogonally invariant states},
  author = {Audenaert, K. and De Moor, B. and Vollbrecht, K. G. H. and Werner, R. F.},
  journal = {Physical Review A},
  volume = {66},
  issue = {3},
  pages = {032310},
  numpages = {11},
  year = {2002},
  month = {Sep},
  publisher = {American Physical Society},
  doi = {10.1103/PhysRevA.66.032310},
  url = {https://link.aps.org/doi/10.1103/PhysRevA.66.032310}
}

@article{FangFawziFawzi2026,
  author={Fang, Kun and Fawzi, Hamza and Fawzi, Omar},
  journal={IEEE Transactions on Information Theory}, 
  title={Efficient Approximation of Regularized Relative Entropies and Applications}, 
  year={2026},
  volume={72},
  number={4},
  pages={2330-2342},
  doi={10.1109/TIT.2026.3661603}}

@article{BeigiRubboliTomamichel2026,
author = {Beigi, Salman and Rubboli, Roberto and Tomamichel, Marco},
year = {2026},
month = {06},
pages = {},
title = {Additivity of quantum relative entropies as a single-copy criterion},
volume = {407},
journal = {Communications in Mathematical Physics},
doi = {10.1007/s00220-026-05653-x}
}

@article{FangFawziFawziAEP2024,
  author        = {Kun Fang and Hamza Fawzi and Omar Fawzi},
  title         = {Generalized quantum asymptotic equipartition},
  journal       = {arXiv preprint arXiv:2411.04035},
  year          = {2024},
  eprint        = {2411.04035},
  archivePrefix = {arXiv},
  primaryClass  = {quant-ph}
}

@article{JonathanPlenio1999,
   title={Entanglement-Assisted Local Manipulation of Pure Quantum States},
   volume={83},
   ISSN={1079-7114},
   url={http://dx.doi.org/10.1103/PhysRevLett.83.3566},
   DOI={10.1103/physrevlett.83.3566},
   number={17},
   journal={Physical Review Letters},
   publisher={American Physical Society (APS)},
   author={Jonathan, Daniel and Plenio, Martin B.},
   year={1999},
   month=Oct, pages={3566–3569} }

@article{Winter2016,
   title={Tight Uniform Continuity Bounds for Quantum Entropies: Conditional Entropy, Relative Entropy Distance and Energy Constraints},
   volume={347},
   ISSN={1432-0916},
   url={http://dx.doi.org/10.1007/s00220-016-2609-8},
   DOI={10.1007/s00220-016-2609-8},
   number={1},
   journal={Communications in Mathematical Physics},
   publisher={Springer Science and Business Media LLC},
   author={Winter, Andreas},
   year={2016},
   month=Mar, pages={291–313} }

@article{WangWildeSu2020,
  title = {Efficiently Computable Bounds for Magic State Distillation},
  author = {Wang, Xin and Wilde, Mark M. and Su, Yuan},
  journal = {Physical Review Letters},
  volume = {124},
  issue = {9},
  pages = {090505},
  numpages = {7},
  year = {2020},
  month = {Mar},
  publisher = {American Physical Society},
  doi = {10.1103/PhysRevLett.124.090505},
  url = {https://link.aps.org/doi/10.1103/PhysRevLett.124.090505}
}

@article{WangWildeSu2019,
   title={Quantifying the magic of quantum channels},
   volume={21},
   ISSN={1367-2630},
   url={http://dx.doi.org/10.1088/1367-2630/ab451d},
   DOI={10.1088/1367-2630/ab451d},
   number={10},
   journal={New Journal of Physics},
   publisher={IOP Publishing},
   author={Wang, Xin and Wilde, Mark M and Su, Yuan},
   year={2019},
   month=Oct, pages={103002} }

@article{LamiMeleRegula2025,
  title = {Computable Entanglement Cost under Positive Partial Transpose Operations},
  author = {Lami, Ludovico and Mele, Francesco Anna and Regula, Bartosz},
  journal = {Physical Review Letters},
  volume = {134},
  issue = {9},
  pages = {090202},
  numpages = {9},
  year = {2025},
  month = {Mar},
  publisher = {American Physical Society},
  doi = {10.1103/PhysRevLett.134.090202},
  url = {https://link.aps.org/doi/10.1103/PhysRevLett.134.090202}
}

@article{Rains1999,
   title={Bound on distillable entanglement},
   volume={60},
   ISSN={1094-1622},
   url={http://dx.doi.org/10.1103/PhysRevA.60.179},
   DOI={10.1103/physreva.60.179},
   number={1},
   journal={Physical Review A},
   publisher={American Physical Society (APS)},
   author={Rains, E. M.},
   year={1999},
   month=July, pages={179–184} }

@article{WangDuan2017,
   title={Irreversibility of Asymptotic Entanglement Manipulation Under Quantum Operations Completely Preserving Positivity of Partial Transpose},
   volume={119},
   ISSN={1079-7114},
   url={http://dx.doi.org/10.1103/PhysRevLett.119.180506},
   DOI={10.1103/physrevlett.119.180506},
   number={18},
   journal={Physical Review Letters},
   publisher={American Physical Society (APS)},
   author={Wang, Xin and Duan, Runyao},
   year={2017},
   pages = {180506},
   month=Nov }

@article{LamiRegulaStreltsov2024,
  author  = {Lami, Ludovico and Regula, Bartosz and Streltsov, Alexander},
  title   = {No-go theorem for entanglement distillation using catalysis},
  journal = {Physical Review A},
  volume  = {109},
  pages   = {L050401},
  year    = {2024},
  doi     = {10.1103/PhysRevA.109.L050401}
}

@article{VedralEtAl1997,
  author  = {V. Vedral and M. B. Plenio and M. A. Rippin and P. L. Knight},
  title   = {Quantifying Entanglement},
  journal = {Physical Review Letters},
  volume  = {78},
  pages   = {2275--2279},
  year    = {1997},
  doi     = {10.1103/PhysRevLett.78.2275}
}

@article{VedralPlenio1998,
  author  = {V. Vedral and M. B. Plenio},
  title   = {Entanglement Measures and Purification Procedures},
  journal = {Physical Review A},
  volume  = {57},
  pages   = {1619--1633},
  year    = {1998},
  doi     = {10.1103/PhysRevA.57.1619}
}

@article{Vidal2000,
  author  = {Guifre Vidal},
  title   = {Entanglement Monotones},
  journal = {Journal of Modern Optics},
  volume  = {47},
  pages   = {355},
  year    = {2000},
  doi     = {10.1080/09500340008244048}
}

@article{HorodeckiLimits2000,
  author  = {Micha{\l} Horodecki and Pawe{\l} Horodecki and Ryszard Horodecki},
  title   = {Limits for Entanglement Measures},
  journal = {Physical Review Letters},
  volume  = {84},
  pages   = {2014},
  year    = {2000},
  doi     = {10.1103/PhysRevLett.84.2014}
}

@article{DonaldHorodeckiRudolph2002,
  author  = {Matthew J. Donald and Micha{\l} Horodecki and Oliver Rudolph},
  title   = {The Uniqueness Theorem for Entanglement Measures},
  journal = {Jounal of Mathematical Physics},
  volume  = {43},
  pages   = {4252--4272},
  year    = {2002},
  doi     = {10.1063/1.1495917}
}

@article{BrandaoGour2015,
  author  = {Fernando G. S. L. Brand{\~a}o and Gilad Gour},
  title   = {Reversible Framework for Quantum Resource Theories},
  journal = {Physical Review Letters},
  volume  = {115},
  pages   = {070503},
  year    = {2015},
  doi     = {10.1103/PhysRevLett.115.070503}
}

@article{KondraDattaStreltsov2021,
  author  = {Tulja Varun Kondra and Chandan Datta and Alexander Streltsov},
  title   = {Catalytic Transformations of Pure Entangled States},
  journal = {Physical Review Letters},
  volume  = {127},
  pages   = {150503},
  year    = {2021},
  doi     = {10.1103/PhysRevLett.127.150503}
}

@article{EisertAudenaertPlenio2003,
   title={Remarks on entanglement measures and non-local state distinguishability},
   volume={36},
   ISSN={1361-6447},
   url={http://dx.doi.org/10.1088/0305-4470/36/20/316},
   DOI={10.1088/0305-4470/36/20/316},
   number={20},
   journal={Journal of Physics A: Mathematical and General},
   publisher={IOP Publishing},
   author={Eisert, J and Audenaert, K and Plenio, M B},
   year={2003},
   month=May, pages={5605–5615} }

@article{LamiRegula2023,
  author        = {Ludovico Lami and Bartosz Regula},
  title         = {No second law of entanglement manipulation after all},
  journal       = {Nature Physics},
  volume        = {19},
  pages         = {184--189},
  year          = {2023},
  doi           = {10.1038/s41567-022-01873-9},
  eprint        = {2111.02438},
  archivePrefix = {arXiv},
  primaryClass  = {quant-ph}
}

@article{RegulaLami2024,
  author        = {Bartosz Regula and Ludovico Lami},
  title         = {Reversibility of quantum resources through probabilistic protocols},
  journal       = {Nature Communications},
  volume        = {15},
  pages         = {3096},
  year          = {2024},
  doi           = {10.1038/s41467-024-47243-2},
  eprint        = {2309.07206},
  archivePrefix = {arXiv},
  primaryClass  = {quant-ph}
}

@article{LamiRegulaDNE2024,
  author        = {Ludovico Lami and Bartosz Regula},
  title         = {Distillable entanglement under dually non-entangling operations},
  journal       = {Nature Communications},
  volume        = {15},
  pages         = {10120},
  year          = {2024},
  doi           = {10.1038/s41467-024-54201-5},
  eprint        = {2307.11008},
  archivePrefix = {arXiv},
  primaryClass  = {quant-ph}
}

@article{BrandaoPlenio2008,
   title={Entanglement theory and the second law of thermodynamics},
   volume={4},
   ISSN={1745-2481},
   url={http://dx.doi.org/10.1038/nphys1100},
   DOI={10.1038/nphys1100},
   number={11},
   journal={Nature Physics},
   publisher={Springer Science and Business Media LLC},
   author={Brandão, Fernando G. S. L. and Plenio, Martin B.},
   year={2008},
   month=Oct, pages={873–877} }

@article{BrandaoPlenio2010,
   title={A Reversible Theory of Entanglement and its Relation to the Second Law},
   volume={295},
   ISSN={1432-0916},
   url={http://dx.doi.org/10.1007/s00220-010-1003-1},
   DOI={10.1007/s00220-010-1003-1},
   number={3},
   journal={Communications in Mathematical Physics},
   publisher={Springer Science and Business Media LLC},
   author={Brandão, Fernando G. S. L. and Plenio, Martin B.},
   year={2010},
   month=Feb, pages={829–851} }

@article{BrandaoPlenioStein2010,
   title={A Generalization of Quantum Stein’s Lemma},
   volume={295},
   ISSN={1432-0916},
   url={http://dx.doi.org/10.1007/s00220-010-1005-z},
   DOI={10.1007/s00220-010-1005-z},
   number={3},
   journal={Communications in Mathematical Physics},
   publisher={Springer Science and Business Media LLC},
   author={Brandão, Fernando G. S. L. and Plenio, Martin B.},
   year={2010},
   month=Feb, pages={791–828} }

@article{HayashiYamasaki2025,
  title={The generalized quantum Stein’s lemma and the second law of quantum resource theories},
  author={Hayashi, Masahito and Yamasaki, Hayata},
  journal={Nature Physics},
  volume={21},
  number={12},
  pages={1988--1993},
  year={2025},
  publisher={Nature Publishing Group UK London},
  doi={10.1038/s41567-025-03047-9}
}

@article{BertaEtAl2023,
   title={On a gap in the proof of the generalised quantum Stein's lemma and its consequences for the reversibility of quantum resources},
   volume={7},
   ISSN={2521-327X},
   url={http://dx.doi.org/10.22331/q-2023-09-07-1103},
   DOI={10.22331/q-2023-09-07-1103},
   journal={Quantum},
   publisher={Verein zur Forderung des Open Access Publizierens in den Quantenwissenschaften},
   author={Berta, Mario and Brandão, Fernando G. S. L. and Gour, Gilad and Lami, Ludovico and Plenio, Martin B. and Regula, Bartosz and Tomamichel, Marco},
   year={2023},
   month=Sept, pages={1103} }

\clearpage
\newpage
\onecolumngrid
\setcounter{equation}{0}
\setcounter{lemma}{0}
\renewcommand{\theequation}{S\arabic{equation}}
\renewcommand{\thelemma}{S\arabic{lemma}}

\begin{center}
\large\textbf{Supplemental material}
\end{center}

\section{Proof of Lemma~\ref{lem:superadditivity}}

For every bipartite system \(X:Y\), define the \(n\)-copy comparison set
\begin{equation}
 \cF_n(X:Y):=\cF(X^n:Y^n),
 \qquad n\geq1.
 \label{eq:F-n-copy}
\end{equation}
Consider the composite hypothesis-testing problem
\begin{equation}
 H_0:\ \xi_{XY}^{\otimes n},
 \qquad
 H_1:\ \tau_n\in\cF_n(X:Y).
 \label{eq:composite-hypotheses}
\end{equation}
For a test \(0\leq T_n\leq I\), the outcome corresponding to \(T_n\) means that \(H_0\) is accepted. Its type-I error and worst-case type-II error are
\begin{align}
 \alpha_n\!\left(T_n\middle|\xi_{XY}\right)
 &:=
 \Tr\!\left[(I-T_n)\xi_{XY}^{\otimes n}\right],
 \label{eq:type-I-definition}\\
 \beta_n\!\left(T_n\middle|\cF_n(X:Y)\right)
 &:=
 \sup_{\tau_n\in\cF_n(X:Y)}
 \Tr(T_n\tau_n) \notag\\
 &=
 h_{\cF_n(X:Y)}(T_n),
 \label{eq:type-II-definition}
\end{align}
where
\begin{equation}
 h_{\mathcal C}(M)
 :=
 \sup_{\tau\in\mathcal C}\Tr(M\tau)
 \label{eq:support-function-proof}
\end{equation}
is the support function of \(\mathcal C\). Thus, \(\alpha_n\) is the probability of rejecting \(H_0\) when \(H_0\) is true, whereas \(\beta_n\) is the largest probability of accepting \(H_0\) over all alternatives in \(\cF_n(X:Y)\).

We first verify that the generalized quantum Stein lemma \cite{BrandaoPlenioStein2010,BertaEtAl2023,HayashiYamasaki2025, FangFawziFawziAEP2024} is applicable. For any fixed bipartition \(X:Y\), conditions \textnormal{(iii)} and \textnormal{(iv)}, applied respectively to \(X^m:Y^m\) and \(X^n:Y^n\), give
\begin{align}
 \cF_m(X:Y)\otimes\cF_n(X:Y)
 &\subseteq
 \cF_{m+n}(X:Y),
 \label{eq:F-internal-tensor}\\
 \cF_m(X:Y)^\circ\otimes\cF_n(X:Y)^\circ
 &\subseteq
 \cF_{m+n}(X:Y)^\circ.
 \label{eq:F-internal-polar}
\end{align}
By the standing assumptions, every \(\cF_n(X:Y)\) is compact and convex, while condition~\textnormal{(i)} ensures that \(\cF_1(X:Y)\) contains a full-rank state. The latter implies the support condition
\begin{equation}
 \supp\xi_{XY}\subseteq\supp\tau_{XY}
\end{equation}
for some \(\tau_{XY}\in\cF_1(X:Y)\) and every state \(\xi_{XY}\). Finally, condition~\textnormal{(ii)} gives invariance of \(\cF_n(X:Y)\) under permutations of the \(n\) copies. Consequently, the sequence \(\{\cF_n(X:Y)\}_{n\geq1}\) satisfies Assumptions~1--3 of Ref.~\cite{BeigiRubboliTomamichel2026}.

The generalized quantum Stein lemma \cite{BrandaoPlenioStein2010,BertaEtAl2023,HayashiYamasaki2025, FangFawziFawziAEP2024}, in the form stated as Theorem~5 of Ref.~\cite{BeigiRubboliTomamichel2026} and proved in Ref.~\cite{FangFawziFawziAEP2024}, therefore gives
\begin{equation}
\begin{aligned}
 D_{\cF(X:Y)}^\infty(\xi_{XY})
 =
 \sup\biggl\{r\geq0:\;&
 \exists\,\{T_n\}_{n\geq1},\quad 0\leq T_n\leq I,\\
 &\alpha_n(T_n\mid\xi_{XY})\longrightarrow0,\\
 &\liminf_{n\to\infty}
 -\frac1n
 \log_2\beta_n\!\left(
 T_n\middle|\cF_n(X:Y)
 \right)
 \geq r
 \biggr\}.
\end{aligned}
\label{eq:generalized-Stein-form}
\end{equation}

We shall also use the support-function consequence of condition~\textnormal{(iv)}. By Ref.~\cite[Lemma~1, Eq.~(21)]{BeigiRubboliTomamichel2026}, the inclusion
\begin{equation}
 \cF(X:Y)^\circ\otimes\cF(X':Y')^\circ
 \subseteq
 \cF(XX':YY')^\circ
\end{equation}
is equivalent to
\begin{equation}
 h_{\cF(XX':YY')}(M\otimes N)
 \leq
 h_{\cF(X:Y)}(M)\,
 h_{\cF(X':Y')}(N)
 \label{eq:support-submultiplicativity-proof}
\end{equation}
for all \(M,N\geq0\).

Now let
\(\omega_{A_1A_2B_1B_2}\) be arbitrary, and denote its marginals by
\[
 \omega_i:=\omega_{A_iB_i},
 \qquad i=1,2.
\]
Set
\begin{equation}
 r_i
 :=
 D_{\cF(A_i:B_i)}^\infty(\omega_i),
 \qquad i=1,2,
\end{equation}
and fix \(\delta>0\). Applying \eqref{eq:generalized-Stein-form} to the two marginal bipartitions, we can choose tests \(0\leq T_{i,n}\leq I\) such that
\begin{equation}
 \alpha_n(T_{i,n}\mid\omega_i)\longrightarrow0,
 \qquad i=1,2,
 \label{eq:marginal-type-I-proof}
\end{equation}
and
\begin{equation}
 \liminf_{n\to\infty}
 -\frac1n
 \log_2\beta_n\!\left(
 T_{i,n}\middle|\cF_n(A_i:B_i)
 \right)
 \geq r_i-\delta,
 \qquad i=1,2.
 \label{eq:marginal-type-II-proof}
\end{equation}

After canonically regrouping tensor factors as
\begin{equation}
 (A_1A_2)^n:(B_1B_2)^n
 \cong
 A_1^nA_2^n:B_1^nB_2^n,
\end{equation}
consider the product test
\begin{equation}
 T_n:=T_{1,n}\otimes T_{2,n}.
 \label{eq:product-test-proof}
\end{equation}
Since
\begin{align}
 I-T_{1,n}\otimes T_{2,n}
 &=
 (I-T_{1,n})\otimes I
 +
 T_{1,n}\otimes(I-T_{2,n}) \notag\\
 &\leq
 (I-T_{1,n})\otimes I
 +
 I\otimes(I-T_{2,n}),
 \label{eq:operator-union-bound-proof}
\end{align}
and the corresponding marginals of \(\omega^{\otimes n}\) are \(\omega_1^{\otimes n}\) and \(\omega_2^{\otimes n}\), we obtain
\begin{align}
 \alpha_n(T_n\mid\omega)
 &=
 \Tr\!\left[(I-T_n)\omega^{\otimes n}\right] \notag\\
 &\leq
 \Tr\!\left[
 ((I-T_{1,n})\otimes I)\omega^{\otimes n}
 \right] \notag\\
 &\quad+
 \Tr\!\left[
 (I\otimes(I-T_{2,n}))\omega^{\otimes n}
 \right] \notag\\
 &=
 \alpha_n(T_{1,n}\mid\omega_1)
 +
 \alpha_n(T_{2,n}\mid\omega_2)
 \longrightarrow0.
 \label{eq:joint-type-I-proof}
\end{align}
Thus, \(\{T_n\}_{n\geq1}\) is an admissible test sequence for the joint null hypothesis \(\omega^{\otimes n}\). Notice that this step does not require \(\omega\) to be a product state.

On the alternative side, condition~\textnormal{(iv)} applied to \(A_1^n:B_1^n\) and \(A_2^n:B_2^n\), together with
\eqref{eq:support-submultiplicativity-proof}, gives
\begin{align}
 &\beta_n\!\left(
 T_n\middle|\cF_n(A_1A_2:B_1B_2)
 \right) \notag\\
 &\quad=
 h_{\cF(A_1^nA_2^n:B_1^nB_2^n)}
 \left(T_{1,n}\otimes T_{2,n}\right) \notag\\
 &\quad\leq
 h_{\cF(A_1^n:B_1^n)}(T_{1,n})\,
 h_{\cF(A_2^n:B_2^n)}(T_{2,n}) \notag\\
 &\quad=
 \beta_n\!\left(
 T_{1,n}\middle|\cF_n(A_1:B_1)
 \right)
 \beta_n\!\left(
 T_{2,n}\middle|\cF_n(A_2:B_2)
 \right).
 \label{eq:joint-type-II-proof}
\end{align}
It follows that
\begin{align}
 &\liminf_{n\to\infty}
 -\frac1n
 \log_2\beta_n\!\left(
 T_n\middle|\cF_n(A_1A_2:B_1B_2)
 \right) \notag\\
 &\quad\geq
 \liminf_{n\to\infty}
 \left[
 -\frac1n
 \log_2\beta_n\!\left(
 T_{1,n}\middle|\cF_n(A_1:B_1)
 \right)
 \right] \notag\\
 &\qquad+
 \liminf_{n\to\infty}
 \left[
 -\frac1n
 \log_2\beta_n\!\left(
 T_{2,n}\middle|\cF_n(A_2:B_2)
 \right)
 \right] \notag\\
 &\quad\geq r_1+r_2-2\delta.
 \label{eq:joint-exponent-proof}
\end{align}

Finally, applying the Stein characterization
\eqref{eq:generalized-Stein-form} to the bipartition
\(A_1A_2:B_1B_2\), and using
\eqref{eq:joint-type-I-proof} and
\eqref{eq:joint-exponent-proof}, yields
\begin{align}
 D_{\cF(A_1A_2:B_1B_2)}^\infty(\omega)
 &\geq r_1+r_2-2\delta \notag\\
 &=
 D_{\cF(A_1:B_1)}^\infty(\omega_1)
 +
 D_{\cF(A_2:B_2)}^\infty(\omega_2)
 -2\delta.
\end{align}
Since \(\delta>0\) is arbitrary,
\begin{equation}
 D_{\cF(A_1A_2:B_1B_2)}^\infty(\omega)
 \geq
 D_{\cF(A_1:B_1)}^\infty(\omega_1)
 +
 D_{\cF(A_2:B_2)}^\infty(\omega_2),
\end{equation}
which proves the lemma. 

\section{Proof of Theorem~\ref{thm:reg-ppt-k-corr-catal}}

For \(k=1\), the Rains sets satisfy the assumptions of Lemma~\ref{lem:superadditivity} by Ref.~\cite[Example~3]{BeigiRubboliTomamichel2026}. For every \(k\geq2\), compactness, convexity, permutation invariance, and the tensor and positive-polar tensor conditions follow from Ref.~\cite[Lemma~25]{FangFawziFawzi2026}, while the maximally mixed state belongs to \(\PPT(A:B)\subseteq\PPT_k(A:B)\) by Ref.~\cite[Proposition~S4]{WangJingZhu2025}. Thus, Lemma~\ref{lem:superadditivity} proves strong superadditivity of \(D_{\PPT_k}^{\infty}\) for every finite \(k\).

Tensor-product additivity follows from Ref.~\cite[Proposition~9]{BeigiRubboliTomamichel2026}. For \(k=1\), the required Assumption~4 is verified in Ref.~\cite[Example~3]{BeigiRubboliTomamichel2026}. For \(k\geq2\), Assumption~4 follows from the tensor and positive-polar tensor constructions of Ref.~\cite[Lemma~25]{FangFawziFawzi2026}, together with permutation invariance of the defining constraints and the fact that \(\PPT_k\) contains the maximally mixed state.

Monotonicity under completely PPT-preserving channels follows from Ref.~\cite[Proposition~S13]{LamiMeleRegula2025} and data processing of the relative entropy; regularization preserves this monotonicity. Asymptotic continuity follows by applying the continuity argument of Ref.~\cite[Lemma~7]{Winter2016} to \(D_{\PPT_k}\) and repeating the telescoping argument in the proof of Ref.~\cite[Corollary~8]{Winter2016}. As in that proof, the variation required in the telescoping step is bounded uniformly in the spectator system by monotonicity under partial trace, tensor-product subadditivity, and the fact that \(\PPT_k\) contains the maximally mixed state.

Combining monotonicity, tensor-product additivity, strong superadditivity, and lower semicontinuity with the correlated-catalysis criterion of Ref.~\cite[Sec.~IV, Lemma~6]{FundamentalLimitsCorrelatedCatalysis} proves Theorem~\ref{thm:reg-ppt-k-corr-catal}.

\section{Proof of Theorem 2}

Define
\begin{equation}
\PPT_\infty(A:B):=\bigcap_{k\geq1}\PPT_k(A:B).
\label{eq:ppt-infty}
\end{equation}
For $k\geq0$ and $X\geq0$, define~\cite{LamiMeleRegula2025}
\begin{equation}
\chi_k(X)
:=
\min\left\{
\Tr S_k:
-S_0\leq X^{\Gamma_B}\leq S_0,\ 
-S_i\leq S_{i-1}^{\Gamma_B}\leq S_i,\ 
1\leq i\leq k
\right\}.
\label{eq:chi-def}
\end{equation}
Unrolling the recursive definition of $\PPT_k$ gives
\begin{equation}
\PPT_k(A:B)
=
\{\tau\geq0:\chi_{k-1}(\tau)\leq1\}.
\label{eq:ppt-chi}
\end{equation}
The quantities $\chi_p$ satisfy
\begin{equation}
\chi_k(tX)=t\chi_k(X),
\qquad t\geq0,
\label{eq:chi-scaling}
\end{equation}
and are nondecreasing in $k$~\cite{LamiMeleRegula2025},
\begin{equation}
\chi_k(X)\leq\chi_{k+1}(X).
\label{eq:chi-monotone}
\end{equation}
Let \(\EcPPT(\sigma)\) denote the zero-error asymptotic entanglement cost of \(\sigma\) under completely PPT-preserving channels.

Moreover, for every bipartite state $\sigma$,
Ref.~\cite[Theorem~1]{LamiMeleRegula2025} shows that
\begin{equation}
\lim_{k\to\infty}\log_2\chi_k(\sigma)
=
\EcPPT(\sigma),
\end{equation}
or equivalently,
\begin{equation}
\lim_{k\to\infty}\chi_k(\sigma)
=
2^{\EcPPT(\sigma)}.
\label{eq:chi-limit}
\end{equation}

For any nonzero $\tau\geq0$, write
\begin{equation}
\tau=t\sigma,
\qquad
t=\Tr\tau,
\qquad
\sigma\in\mathcal D(AB).
\end{equation}
By Eqs.~\eqref{eq:ppt-chi} and \eqref{eq:chi-scaling},
$\tau\in\PPT_\infty(A:B)$ if and only if
\begin{equation}
t\chi_k(\sigma)\leq1
\qquad
\text{for every }k\geq0.
\end{equation}
Since $\chi_p(\sigma)$ is nondecreasing, Eq.~\eqref{eq:chi-limit}
shows that this is equivalent to
\begin{equation}
t\leq2^{-\EcPPT(\sigma)}.
\end{equation}
Hence
\begin{equation}
\PPT_\infty(A:B)
=
\left\{
t\sigma:\ 
\sigma\in\mathcal D(AB),\
0\leq t\leq2^{-\EcPPT(\sigma)}
\right\}.
\label{eq:ppt-infty-cost}
\end{equation}

We next relate $\PPT_\infty$ to the regularized PPT-relative entropy.

\begin{lemma}[Exact-cost domination]
\label{lem:exact-cost-domination}
For all bipartite states $\rho$ and $\sigma$,
\begin{equation}
D_{\PPT}^{\infty}(\rho)
\leq
D(\rho\Vert\sigma)+\EcPPT(\sigma).
\label{eq:cost-domination}
\end{equation}
\end{lemma}

\begin{proof}
The claim is trivial if $D(\rho\Vert\sigma)=\infty$. Otherwise, fix
$R>\EcPPT(\sigma)$. Let
\begin{equation}
\Phi_M:=|\Phi_M\rangle\!\langle\Phi_M|,
\qquad
|\Phi_M\rangle:=\frac{1}{\sqrt M}\sum_{i=1}^M|ii\rangle
\end{equation}
denote the maximally entangled state of Schmidt rank $M$.
By the definition of the exact PPT entanglement cost, for all
sufficiently large $n$ there exist a completely PPT-preserving channel
$\Lambda_n$ and an integer $M_n$ satisfying
\begin{equation}
\log_2M_n\leq nR,
\qquad
\Lambda_n(\Phi_{M_n})=\sigma^{\otimes n}.
\end{equation}
Define
\begin{equation}
\zeta_M:=\frac1M\sum_{i=1}^M|ii\rangle\!\langle ii|.
\end{equation}
Since $\zeta_M$ is separable and
\begin{equation}
\zeta_M\geq\frac1M\Phi_M,
\end{equation}
the state
\begin{equation}
\omega_n:=\Lambda_n(\zeta_{M_n})
\end{equation}
is PPT and satisfies
\begin{equation}
\omega_n\geq M_n^{-1}\sigma^{\otimes n}.
\end{equation}
Operator monotonicity of the logarithm gives
\begin{equation}
\begin{aligned}
D_{\PPT(A^n:B^n)}(\rho^{\otimes n})
&\leq D(\rho^{\otimes n}\Vert\omega_n)\\
&\leq D(\rho^{\otimes n}\Vert\sigma^{\otimes n})
+\log_2M_n\\
&\leq n\bigl(D(\rho\Vert\sigma)+R\bigr).
\end{aligned}
\end{equation}
Dividing by $n$ and taking $n\to\infty$ gives
\begin{equation}
D_{\PPT}^{\infty}(\rho)
\leq D(\rho\Vert\sigma)+R.
\end{equation}
Taking the infimum over all $R>\EcPPT(\sigma)$ proves
Eq.~\eqref{eq:cost-domination}.
\end{proof}

Using Eq.~\eqref{eq:ppt-infty-cost} and
\begin{equation}
D(\rho\Vert t\sigma)
=
D(\rho\Vert\sigma)-\log_2t,
\qquad t>0,
\end{equation}
we obtain
\begin{equation}
\begin{aligned}
D_{\PPT_\infty(A:B)}(\rho)
&=
\min_{\sigma\in\mathcal D(AB)}
\min_{0<t\leq2^{-\EcPPT(\sigma)}}
\left\{
D(\rho\Vert\sigma)-\log_2t
\right\}\\
&=
\min_{\sigma\in\mathcal D(AB)}
\left\{
D(\rho\Vert\sigma)+\EcPPT(\sigma)
\right\}.
\end{aligned}
\label{eq:inf-convolution}
\end{equation}
Lemma~\ref{lem:exact-cost-domination}, applied to $\rho^{\otimes n}$,
therefore gives
\begin{equation}
D_{\PPT_\infty(A^n:B^n)}(\rho^{\otimes n})
\geq
D_{\PPT}^{\infty}(\rho^{\otimes n})
=
nD_{\PPT}^{\infty}(\rho).
\end{equation}
Dividing by $n$ and taking the regularization yields
\begin{equation}
D_{\PPT_\infty}^{\infty}(\rho)
\geq
D_{\PPT}^{\infty}(\rho).
\end{equation}
Conversely, since
\begin{equation}
\PPT(A^n:B^n)\subseteq\PPT_\infty(A^n:B^n)
\end{equation}
for every $n$,
\begin{equation}
D_{\PPT_\infty}^{\infty}(\rho)
\leq
D_{\PPT}^{\infty}(\rho).
\end{equation}
Thus
\begin{equation}
D_{\PPT_\infty}^{\infty}(\rho)
=
D_{\PPT}^{\infty}(\rho).
\label{eq:ppt-infty-equality}
\end{equation}

It remains to verify the tensor conditions for \(\PPT_\infty\). By Eq.~\eqref{eq:ppt-infty-cost}, tensor closure follows from the full tensor additivity of the exact PPT entanglement cost~\cite{LamiMeleRegula2025}. If
\begin{equation}
t\sigma\in\PPT_\infty(A:B),
\qquad
s\eta\in\PPT_\infty(A':B'),
\end{equation}
then
\begin{equation}
ts
\leq
2^{-\EcPPT(\sigma)-\EcPPT(\eta)}
=
2^{-\EcPPT(\sigma\otimes\eta)},
\end{equation}
and hence
\begin{equation}
(t\sigma)\otimes(s\eta)
\in
\PPT_\infty(AA':BB').
\end{equation}

For the positive-polar condition, recall that $\PPT_\infty=\bigcap_{k\geq1}\PPT_k$. Since the $\PPT_k$ form a decreasing sequence of compact sets, for every $X\geq0$,
\begin{equation}
h_{\PPT_\infty(A:B)}(X)
=
\lim_{k\to\infty}h_{\PPT_k(A:B)}(X).
\label{eq:ppt-infty-support}
\end{equation}
Indeed, let $\tau_k\in\PPT_k(A:B)$ attain $h_{\PPT_k(A:B)}(X)$. By compactness of $\PPT_1(A:B)$, a subsequence $\tau_{k_j}$ converges to some $\tau_*$. For every fixed $r$, $\tau_{k_j}\in\PPT_r(A:B)$ for all sufficiently large $j$; hence, since $\PPT_r(A:B)$ is closed, $\tau_*\in\PPT_r(A:B)$. Thus $\tau_*\in\PPT_\infty(A:B)$, which proves Eq.~\eqref{eq:ppt-infty-support}.

Using the finite-level positive-polar tensor condition and
Eq.~\eqref{eq:ppt-infty-support}, for $X,Y\geq0$ we obtain
\begin{equation}
\begin{aligned}
h_{\PPT_\infty(AA':BB')}(X\otimes Y)
&=
\lim_{k\to\infty}
h_{\PPT_k(AA':BB')}(X\otimes Y)\\
&\leq
\lim_{k\to\infty}
h_{\PPT_k(A:B)}(X)
h_{\PPT_k(A':B')}(Y)\\
&=
h_{\PPT_\infty(A:B)}(X)
h_{\PPT_\infty(A':B')}(Y).
\end{aligned}
\end{equation}
By Ref.~\cite[Lemma~1]{BeigiRubboliTomamichel2026},
the above support-function inequality is equivalent to
\begin{equation}
\PPT_\infty(A:B)^\circ\otimes
\PPT_\infty(A':B')^\circ
\subseteq
\PPT_\infty(AA':BB')^\circ.
\end{equation}
The remaining assumptions of Lemma~\ref{lem:superadditivity} are immediate:
$\PPT_\infty$ is compact and convex as an intersection of the
$\PPT_k$, contains every PPT state and hence a full-rank state, and
is permutation invariant. Lemma~\ref{lem:superadditivity} therefore implies that
$D_{\PPT_\infty}^{\infty}$ is strongly superadditive.

Tensor-product additivity follows from Ref.~\cite[Proposition~9]{BeigiRubboliTomamichel2026}. To apply that result, consider finite-dimensional bipartite systems $A_j:B_j$, $j=1,\ldots,r$, and the family
\begin{equation}
\mathcal F_{\vec m}
:=
\PPT_\infty\!\left(
\bigotimes_{j=1}^r A_j^{m_j}:
\bigotimes_{j=1}^r B_j^{m_j}
\right),
\qquad
\vec m=(m_1,\ldots,m_r).
\end{equation}
We verify Assumption~4 of Ref.~\cite{BeigiRubboliTomamichel2026}.
Assumption~4.A follows from compactness and convexity of
$\PPT_\infty$. Assumptions~4.B and~4.C follow by repeated application
of the tensor closure and positive-polar tensor closure established
above. Assumption~4.D holds because the maximally mixed state belongs
to $\PPT\subseteq\PPT_\infty$ and is full rank. Finally,
Assumption~4.E follows from permutation invariance of
$\PPT_\infty$. Hence the complete Assumption~4 is satisfied, and
Ref.~\cite[Proposition~9]{BeigiRubboliTomamichel2026} gives
\begin{equation}
D_{\PPT_\infty(A_1A_2:B_1B_2)}^\infty(\rho\otimes\sigma)
=
D_{\PPT_\infty(A_1:B_1)}^\infty(\rho)
+
D_{\PPT_\infty(A_2:B_2)}^\infty(\sigma).
\label{eq:ppt-infty-additivity}
\end{equation}

Together with ordinary monotonicity under completely PPT-preserving maps and lower semicontinuity, the tensor-product additivity and strong superadditivity established above imply, by Ref.~\cite[Lemma~6]{FundamentalLimitsCorrelatedCatalysis}, that $D_{\PPT}^{\infty}$ is monotone under correlated catalytic transformations implemented by completely PPT-preserving maps.
This completes the proof.

\section{Proof of Theorem 3}

It was shown in Ref.~\cite{LamiRegulaStreltsov2024} that any normalized,
additive, strongly superadditive, and asymptotically continuous monotone
\(M\) yields the correlated-catalytic rate bounds
\begin{equation}
E_{d,\PPT}^{\rm cc}(\rho)
\leq M(\rho)
\leq E_{c,\PPT}^{\rm cc}(\rho).
\label{eq:cc-general-rate-bounds}
\end{equation}
Here
\begin{equation}
E_{d,\PPT}^{\rm cc}(\rho)
:=
\sup\left\{
R\geq0:
\begin{array}{l}
\exists\,\tau_n\ \text{and a completely PPT-preserving channel }\Lambda_n
\text{ such that}\\[1mm]
\omega_n=\Lambda_n(\rho^{\otimes n}\otimes\tau_n),\qquad
\omega_{n,\rm cat}=\tau_n,\\[1mm]
\displaystyle
\frac12\left\|
\omega_{n,\rm sys}
-\Phi_2^{\otimes\lfloor Rn\rfloor}
\right\|_1
\xrightarrow[n\to\infty]{}0
\end{array}
\right\},
\label{eq:cc-distillation-def}
\end{equation}
and
\begin{equation}
E_{c,\PPT}^{\rm cc}(\rho)
:=
\inf\left\{
R\geq0:
\begin{array}{l}
\exists\,\tau_n\ \text{and a completely PPT-preserving channel }\Lambda_n
\text{ such that}\\[1mm]
\omega_n=\Lambda_n(\Phi_2^{\otimes\lceil Rn\rceil}\otimes\tau_n),\qquad
\omega_{n,\rm cat}=\tau_n,\\[1mm]
\displaystyle
\frac12\left\|
\omega_{n,\rm sys}
-\rho^{\otimes n}
\right\|_1
\xrightarrow[n\to\infty]{}0
\end{array}
\right\}.
\label{eq:cc-cost-def}
\end{equation}
Here \(\omega_{n,\rm sys}\) and \(\omega_{n,\rm cat}\) denote the system and
catalyst marginals of the final state, respectively; no product structure
between them is required.

Applying Eq.~\eqref{eq:cc-general-rate-bounds} to
\(D_{\PPT_1}^{\infty}\) for distillation and to
\(D_{\PPT}^{\infty}\) for formation, whose required properties were
established above, gives
\begin{equation}
E_{d,\PPT}^{\rm cc}(\rho)
\leq D_{\PPT_1}^{\infty}(\rho),
\qquad
D_{\PPT}^{\infty}(\rho)
\leq E_{c,\PPT}^{\rm cc}(\rho).
\label{eq:cc-ppt-monotone-bounds}
\end{equation}

For the state \(\rho_v\), Wang and Duan proved~\cite{WangDuan2017}
\begin{equation}
\begin{aligned}
E_{d,\PPT}(\rho_v)
&=
D_{\PPT_1}^{\infty}(\rho_v)
=
\log_2\!\left(1+\frac1{\sqrt2}\right),\\
D_{\PPT}^{\infty}(\rho_v)
&=
E_{c,\PPT}(\rho_v)
=
1.
\end{aligned}
\label{eq:wang-duan-exact-values}
\end{equation}

Since every ordinary PPT protocol is a correlated-catalytic PPT protocol
with a trivial catalyst,
\begin{equation}
E_{d,\PPT}^{\rm cc}(\rho_v)
\geq E_{d,\PPT}(\rho_v),
\qquad
E_{c,\PPT}^{\rm cc}(\rho_v)
\leq E_{c,\PPT}(\rho_v).
\label{eq:trivial-catalyst-bounds}
\end{equation}

Combining Eqs.~\eqref{eq:cc-ppt-monotone-bounds}--\eqref{eq:trivial-catalyst-bounds}
yields
\begin{equation}
E_{d,\PPT}^{\rm cc}(\rho_v)
=
\log_2\!\left(1+\frac1{\sqrt2}\right)
<
1
=
E_{c,\PPT}^{\rm cc}(\rho_v),
\label{eq:cc-ppt-irreversibility-supp}
\end{equation}
which proves the claim.

\section{Application to resource theory of magic}
The lemma \ref{lem:superadditivity} applies beyond entanglement theory, in particular to the
resource theory of magic for systems composed of qudits of odd dimensions
equipped with the discrete Wigner representation. Define the comparison set
\begin{equation}
\mathcal W(A)
:=
\left\{
\tau_A\geq0:
\|\tau_A\|_{W,1}\leq1
\right\},
\end{equation}
where \(\|\cdot\|_{W,1}\) denotes the Wigner trace norm
\cite{WangWildeSu2020}. The corresponding relative-entropy distance is known
as the thauma, with regularization
\begin{equation}
D_{\mathcal W}^{\infty}(\rho)
=
\lim_{n\to\infty}
\frac{1}{n}
D_{\mathcal W}\bigl(\rho^{\otimes n}\bigr)
\end{equation}
\cite{WangWildeSu2020,WangWildeSu2019}.

The sets \(\mathcal W\) satisfy the relevant tensor and positive-polar tensor
conditions \cite[Example~4]{BeigiRubboliTomamichel2026}; hence
Lemma~\ref{lem:superadditivity} implies that
\(D_{\mathcal W}^{\infty}\) is strongly superadditive. Its tensor-product
additivity follows from Ref.~\cite[Proposition~9 and Example~4]
{BeigiRubboliTomamichel2026}, while its monotonicity under completely
positive Wigner-preserving channels follows from
Ref.~\cite{WangWildeSu2019}. Moreover, the sets \(\mathcal W\) are compact
and convex, contain the maximally mixed state, and are closed under tensor
products. Applying Ref.~\cite[Lemma~7]{Winter2016} together with the
telescoping argument in the proof of Ref.~\cite[Corollary~8]{Winter2016}
therefore establishes asymptotic continuity of
\(D_{\mathcal W}^{\infty}\). These properties imply that
\(D_{\mathcal W}^{\infty}\) is monotone under correlated catalytic
transformations implemented by completely positive Wigner-preserving
channels, and hence in particular under correlated catalytic stabilizer
operations. To the best of our knowledge, the strong superadditivity and the
resulting correlated-catalytic monotonicity of the regularized thauma have not been observed previously.

\end{document}